\documentclass[11pt]{article}

\usepackage[margin=1in]{geometry}
\usepackage{fullpage,wrapfig,amsmath, amsthm, amssymb, algorithm,thmtools,algpseudocode,enumerate,caption,framed,paralist,sidecap,xspace}
\usepackage[usenames,dvipsnames,svgnames,table]{xcolor}
\definecolor{darkgreen}{rgb}{0.0,0,0.9}
\usepackage[colorlinks=true,pdfpagemode=UseNone,citecolor=Blue,linkcolor=Red,urlcolor=Black,pagebackref]{hyperref}
\usepackage{authblk}
\usepackage[colorinlistoftodos,prependcaption,textsize=tiny,]{todonotes}

\newtheorem{theorem}{Theorem}[section]
\newtheorem{lemma}{Lemma}[section]

\newcommand{\R}{\mathcal{R}}
\newcommand{\D}{\mathcal{D}}
\newcommand{\F}{\mathcal{F}}
\newcommand{\W}{\mathsf{W[1]}}
\newcommand{\fpt}{\mathsf{FPT}}

\newcommand{\np}{\mathsf{NP}}

\newcommand{\tip}{{2-interval pattern problem}\xspace}
\newcommand{\nst}{\sqsubset}
\newcommand{\btw}{\between}

\title{Parameterized Complexity of Two-Interval Pattern Problem\thanks{This work is supported in part by Natural Sciences and Engineering Research Council of Canada (NSERC).}}

\author[1]{Prosenjit Bose}
\author[1]{Saeed Mehrabi}
\author[2]{Debajyoti Mondal}

\affil[1]{{\small School of Computer Science, Carleton University, Ottawa, Canada.

\texttt{jit@scs.carleton.ca, saeed.mehrabi@carleton.ca}}
}

\affil[1]{{\small Department of Computer Science, University of Saskatchewan, Saskatoon, Canada.

\texttt{d.mondal@usask.ca}}
}

\date{}

\begin{document}

\maketitle

\begin{abstract}
A \emph{2-interval} is the union of two disjoint intervals on the real line. Two 2-intervals $D_1$ and $D_2$ are \emph{disjoint} if their intersection is empty (i.e., no interval of $D_1$ intersects any interval of $D_2$). There can be three different relations between two disjoint 2-intervals; namely, preceding ($<$), nested ($\nst$) and crossing ($\btw$). Two 2-intervals $D_1$ and $D_2$ are called \emph{$R$-comparable} for some $R\in\{<,\nst,\btw\}$, if either $D_1RD_2$ or $D_2RD_1$. A set $\D$ of disjoint 2-intervals is $\R$-comparable, for some $\R\subseteq\{<,\nst,\btw\}$ and $\R\neq\emptyset$, if every pair of 2-intervals in $\R$ are $R$-comparable for some $R\in\R$. Given a set of 2-intervals and some $\R\subseteq\{<,\nst,\btw\}$, the objective of the \emph{\tip} is to find a largest subset of 2-intervals that is $\R$-comparable.

The \tip is known to be $\W$-hard when $|\R|=3$ and $\np$-hard when $|\R|=2$ (except for $\R=\{<,\nst\}$, which is solvable in quadratic time). In this paper, we fully settle the parameterized complexity of the problem by showing it to be  $\W$-hard for both $\R=\{\nst,\btw\}$ and $\R=\{<,\btw\}$ (when parameterized by the size of an optimal solution); this answers the open question posed by Vialette [Encyclopedia of Algorithms, 2008].
\end{abstract}

\section{Introduction}
\label{sec:introduction}
Interval graphs and their generalizations are often used to study problems in resource allocation, scheduling, and DNA mapping. In 2002, Vialette~\cite{DBLP:conf/cpm/Vialette02} proposed a geometric description of RNA 
helices in an attempt to improve the  understanding of the computational complexity  for finding structured patterns in RNA sequences. In particular, Vialette modeled the RNA secondary structure using a set of 2-intervals, which inspired subsequent research (e.g., see~\cite{DBLP:reference/algo/Vialette08}) on examining the properties of the geometric graphs arising from such representations. The \emph{\tip}, introduced by Vialette~\cite{DBLP:journals/tcs/Vialette04}, is a widely studied pattern, and the main topic of this paper.

A \emph{2-interval} is the union of two disjoint intervals on the real line. Two 2-intervals $D_1$ and $D_2$ are \emph{disjoint} if their intersection is empty; that is, no interval of $D_1$ intersects any interval of $D_2$. We can define three different relations between two disjoint 2-intervals: one 2-interval lies entirely to the left of the other one (called \emph{preceding} and denoted by $<$), one 2-interval is nested within the other one (called \emph{nested} and denoted by $\nst$), and the intervals of the two 2-intervals alternate on the real line (called \emph{crossing} and denoted by $\btw$). See Figure~\ref{fig:exampleForRelations}(a) for an example; a formal definition is given in Section~\ref{sec:prelimins}. Two 2-intervals $D_1$ and $D_2$ are \emph{$R$-comparable} for some $R\in\{<,\nst,\btw\}$ if either $D_1RD_2$ or $D_2RD_1$. A set $\D$ of disjoint 2-intervals is $\R$-comparable, for some $\R\subseteq\{<,\nst,\btw\}$ and $\R\neq\emptyset$, if every pair of 2-intervals in $\R$ are $R$-comparable for some $R\in\R$. In the \emph{\tip}, we are given a set of 2-intervals and a set $\R\subseteq\{<,\nst,\btw\}$, and the objective is to compute a largest subset of 2-intervals that is $\R$-comparable.   Figure~\ref{fig:exampleForRelations}(b) illustrates such an example.

The \tip can model various scenarios in the context of RNA structure prediction. While looking for certain RNA structures, some common approaches to cope with intractability are either to restrict the class of pseudoknots~\cite{RivasRE99} or to  apply heuristics~\cite{10.1093/nar/28.4.991,RenRCH05,DBLP:journals/bioinformatics/RuanSZ04}. Vialette~\cite{DBLP:journals/tcs/Vialette04}   proposed that one can obtain a relevant set of 2-intervals from an RNA sequence by selecting stable stems, e.g., using a simplified thermodynamic model without accounting for loop energy~\cite{DBLP:journals/bioinformatics/RuanSZ04,DBLP:journals/tcs/Vialette04,DBLP:conf/wabi/ZhaoMC06}. Then, the prediction of the RNA structure is equivalent to finding a maximum subset of non-conflicting (i.e., disjoint) 2-intervals.

\begin{figure}[t]
	\includegraphics[width=.9\textwidth]{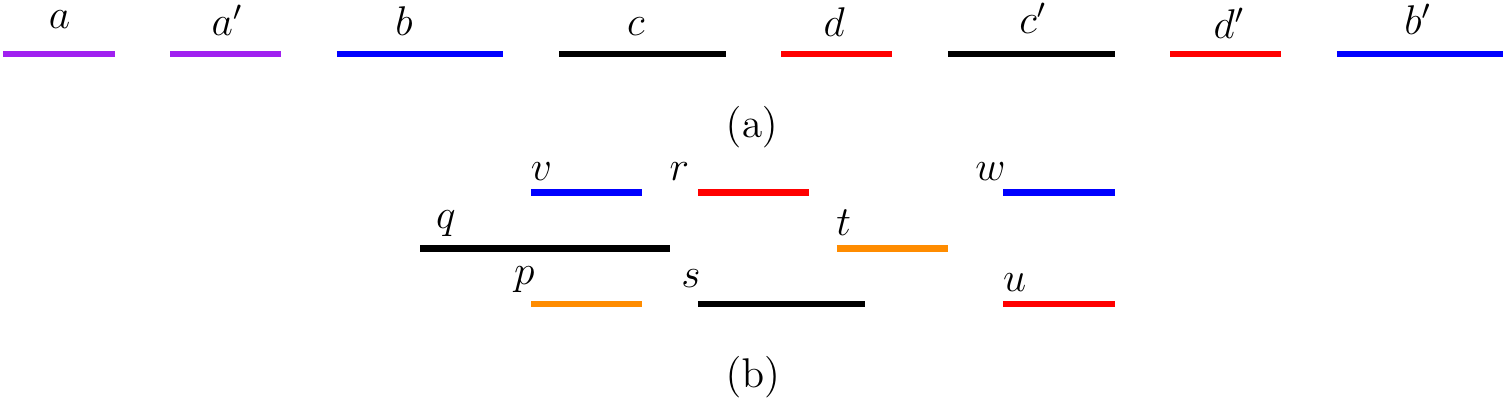}
	\caption{(a) An example for showing the three possible relations between a pair of 2-intervals; here, the same-colour intervals form a 2-interval. Then, $(a,a')<(d,d')$, $(d,d')\nst (b,b')$ and $(c,c')\btw (d,d')$. (b) An instance of \tip with $\R=\{<,\btw\}$ and the 2-intervals are $\{(u,w),(q,s),(r,u),(p,t)\}$. The 2-intervals $\{(p,t),(r,u)\}$ form a  largest subset that is $\R$-comparable.}
	\label{fig:exampleForRelations}
\end{figure}

\paragraph{Related work.} Vialette~\cite{DBLP:journals/tcs/Vialette04} observed that if $|\R|=1$, then the \tip is polynomial-time solvable by reductions to the maximum independent set  problem on interval graphs, or to the maximum clique problem on comparability graphs. The running time of these algorithms have been improved since then, and expressed in terms of the number of input 2-intervals and various  interval-related parameters such as their lengths or overlap~\cite{DBLP:journals/jco/ChenYY07}. For the case when $|\R|$ = 2, the problem is  solvable in polynomial time when $\R=\{<,\nst\}$~\cite{DBLP:journals/tcs/Vialette04}. However, if $\R=\{<,\btw\}$ or  $\R=\{\nst,\btw\}$, then the problem is known to be  $\np$-hard, even if the intervals of every 2-interval have unit length~\cite{DBLP:journals/tcs/Vialette04,DBLP:conf/cpm/BlinFV04}. If $|\R|=3$, i.e., $\R=\{<,\nst,\btw\}$, then the $\np$-hardness of the problem follows from the hardness of recognizing 2-interval graphs~\cite{DBLP:journals/dam/WestS84}.

The approximability of the $\np$-hard models of the \tip was studied by Crochemore et al.~\cite{DBLP:journals/tcs/CrochemoreHLRV08}. They gave polynomial-time algorithms for the problem with approximation factors 4 when $\R=\{<,\nst,\btw\}$ or $\R=\{\nst,\btw\}$, and 6 when $\R=\{<,\nst\}$. They also showed that the results hold for the weighted case, i.e., when each 2-interval is associated with a weight and the goal is to find a maximum weight subset. These factors are improved to 3 when the intervals of every input 2-interval have unit length~\cite{DBLP:journals/tcs/CrochemoreHLRV08}, where they also considered the case when the 2-intervals are weighted. For $\R=\{<,\btw\}$ (and arbitrary input 2-intervals), Jiang~\cite{DBLP:journals/jco/Jiang07} improved the approximation factor to 2 and subsequently to $1+\epsilon$ for any $\epsilon>0$~\cite{DBLP:conf/cocoa/Jiang07}.

The problem is $\W$-hard when $\R=\{<,\nst,\btw\}$, because in this case, the problem is equivalent to computing a maximum independent set on 2-interval graphs, and the latter is known to be $\W$-hard~\cite{DBLP:journals/tcs/FellowsHRV09}; see Section~\ref{sec:hardness} for more details. For $|\R| = 2$, to the best of our knowledge, the only parameterized result is the work of Crochemore et al.~\cite{DBLP:conf/cpm/BlinFV04}  who proved that the problem is fixed-parameter tractable, but only when $\R=\{\nst,\btw\}$, the input intervals all have unit length and the tractability is with respect to the \emph{forward crossing number}: the maximum number of 2-intervals that cross a 2-interval ``from the right''.

\paragraph{Our results.} In this paper, we answer a question of Vialette~\cite{DBLP:reference/algo/Vialette08} by proving that the \tip is $\W$-hard when $\R=\{\nst,\btw\}$ and $\R=\{<,\btw\}$. Our $\W$-hardness result is inspired by the reduction of the $k$-independent set problem used by Fellows et al.~\cite{DBLP:journals/tcs/FellowsHRV09}. Their reduction requires all three relations (i.e., they prove the $\W$-hardness of the \tip when $\R = \{<, \btw, \nst\}$). Prior to our work, it was known that the complexity of the problem is polynomial when $\R=\{<, \nst\}$~\cite{DBLP:journals/jco/ChenYY07}, but it was unknown whether the problem is fixed-parameter tractable (when parameterized by the size of an optimal solution) or it is $\W$-hard for $\R=\{\nst,\btw\}$ and $\R=\{<,\btw\}$. Hence, our $\W$-hardness result fully settles the parameterized complexity of the \tip.

\section{Preliminaries}
\label{sec:prelimins}
In this section, we give some definitions and notation that will be used throughout the paper.

A 2-interval $D$ is the union of two disjoint intervals on the real line; that is, $D=(A,B)$ and the interval $A$ lies to the left of the interval $B$. For a pair of disjoint intervals $I,J$, we write $I<J$ when $I$ is to the left of $J$. Two 2-intervals $D_i$ and $D_j$ are \emph{disjoint} if $(I_i\cup J_i)\cap (I_j\cup J_j)=\emptyset$. Moreover, for two disjoint 2-intervals $D_i$ and $D_j$, we say that $D_i$ is \emph{preceding} (resp., \emph{nested} in, \emph{crossing}) $D_j$ if $I_i<J_i<I_j<J_j$ (resp., $I_j<I_i<J_i<J_j$, $I_i<I_j<J_i<J_j$. We write $D_i<D_j$ (resp., $D_i\nst D_j$, $D_i\btw D_j$) when $D_i$ is preceding (resp., nested in, crossing) $D_j$.

We say that two 2-intervals $D_i$ and $D_j$ are \emph{$R$-comparable}, for some $R\in\{<,\nst,\btw\}$, if (i) $D_i$ and $D_j$ are disjoint and (ii) either $D_i R D_j$ or $D_j R D_i$. Let $S$ be a set of $n$ 2-intervals on the real line, and let $\R\subseteq\{<,\nst,\btw\}$ such that $\R\neq\emptyset$. Then, a set $\D\subseteq S$ is called \emph{$\R$-comparable} if every pair of 2-intervals in $\D$ are $R$-comparable for some $R\in\R$. Given $S$ and some $\R\in\{<,\nst,\btw\}$, the objective of the \tip is to compute a largest subset $\D\subseteq S$ such that $\D$ is $\R$-comparable.

Given a graph $G$ and a parameter $k$, the $k$-independent set problem asks whether there is an independent set of size $k$ in $G$. Fellows et al.~\cite{DBLP:journals/tcs/FellowsHRV09} proved that the $k$-independent set problem is $\W$-hard on 2-interval graphs when $\R = \{<, \btw, \nst\}$. Our $\W$-hardness results are also based on showing reductions from the \emph{$k$-Multicoloured Clique Problem}, which is known to be $\W$-hard~\cite{DBLP:journals/tcs/FellowsHRV09}. The problem is defined as follows.

\begin{framed}
\noindent {\bf Problem:} $k$-Multicoloured Clique.

\noindent {\bf Input:} A graph $G$, and a vertex-colouring $c:V(G)\rightarrow\{1,2,\dots,k\}$ for $G$.

\noindent {\bf Question:} Is there a clique of size $k$ in $G$ such that, for each $c\in\{1,2,\dots,k\}$, there is exactly one vertex in the clique that has colour $c$?
\end{framed}

\section{$\W$-Hardness}
\label{sec:hardness}
In this section, we prove that the \tip is $\W$-hard when $\R=\{\nst,\btw\}$ and $\R=\{<,\btw\}$. Our reduction is inspired by that of Fellows et al.~\cite{DBLP:journals/tcs/FellowsHRV09}. Let $(G,c,k)$ be an instance of the $k$-multicoloured clique problem (we assume w.l.o.g. for our purposes that $c$ is a proper colouring\footnote{Otherwise, one can remove the edges whose end vertices are coloured with the same colour.}). We construct a set $\F$ of 2-intervals such that $G$ has a multicoloured clique of size $k$ if and only if $\F$ contains a set of $k'=2k+4{k\choose 2}$ disjoint 2-intervals that are pairwise comparable in one of the relations in $\R$; the value of $k'$ will be clear from our construction. We first describe an outline of the construction and the corresponding gadgets. Then, we give the details on how to organize the gadgets on the real line specific to each of the sets $\R=\{\nst,\btw\}$ and $\R=\{<,\btw\}$. For a colour $i\in\{1,2,\dots,k\}$, let $V_i(G)$ denote the set of vertices of $G$ that have colour $i$. Moreover, for every distinct pair of colours $i,j$, let $E_{(i,j)}(G)$ denote the set of edges $(u,v)$ of $G$ such that $\{c(u),c(v)\}=\{i,j\}$. That is, $E_{(i,j)}(G)$ consists of all the edges whose end vertices are coloured with two distinct colours $i$ and $j$.

\paragraph{Outline.} The construction consists of two main types of gadgets: \emph{selection} and \emph{validation}. By selection gadgets, we ensure that 2-intervals representing $k$ vertices with distinct colours and ${k\choose 2}$ edges with distinct pairs of colours are selected. By validation gadgets, we ensure that the selected set of 2-intervals are valid in the sense that the $k$ selected vertices are actually adjacent in the graph and the selected edges are indeed over the selected set of vertices. We group the 2-intervals corresponding to vertices of the same colour together in a \emph{vertex-selection} gadget in such a way that any feasible solution for the \tip will have 2-intervals corresponding to one vertex per vertex-selection gadget. Similarly, we group the 2-intervals corresponding to edges with the same pairs of distinct colours $\{i,j\}$ together in a \emph{edge-selection} gadget such that any feasible solution for the \tip will have 2-intervals corresponding to one edge $(u,v)$ with $\{c(u),c(v)\}=\{i,j\}$. We will then organize the gadgets on the real line in such a way that any feasible solution will contain 2-intervals that are  
$\R$-comparable. 

Given $(G,c,k)$, we associate one 2-interval $I_v$ for each vertex $v\in V(G)$. Moreover, we associate four 2-intervals for each edge $(u,v)\in E(G)$: two 2-intervals $I_{(u,v)}$ and $I_{(v,u)}$ for each ``direction'' of the edge and two 2-intervals $I_{\{u,v\}}$ and $I'_{\{u,v\}}$ that are undirected. The 2-intervals for ``directed'' edges will be used for validation, and we will show below how they are constructed. Therefore, the number of 2-intervals of the constructed instance will be $|V(G)|+4|E(G)|$. We next give the details of each type of gadgets.

\paragraph{Vertex-selection gadget.} For each colour $c\in\{1,2,\dots,k\}$, we construct a vertex-selection gadget. The gadget has two components, which we denote  by $I_1(c)$ and $I_2(c)$; see Figure~\ref{fig:vertexGadget} for an illustration. The component $I_1(c)$ has $|V_c(G)|$ ``rows'' of intervals, each of which has $(k+1)$ ``columns''; each row corresponds to a vertex of $G$ with colour $c$. The intervals in the same column pairwise intersect. Moreover, for the intervals in a fixed column $j$, we assign an offset such that each interval in row $i>1$ intersects the interval that is in column $j+1$ and row $i-1$; see Figure~\ref{fig:vertexGadget}. The component $I_2(c)$ consists of two columns of intervals, and each columns has $|V_c(G)|$ rows. Here, we assign an offset such that the interval in the first column and row $i$ intersects the interval in the second column and row $i+1$ (see Figure~\ref{fig:vertexGadget}). 

For each vertex $v\in V_c(G)$, we associate two 2-intervals $I_v$ and $I'_v$ as follows. The first (resp., second) 2-interval $I_v$ (resp., $I'_v$) is composed of the interval in the first (resp., last) column of $I_1(c)$ that corresponds to $v$ and the interval in the first (resp., second) column of $I_2(c)$ that corresponds to $v$. These 2-intervals are illustrated with dashed lines in Figure~\ref{fig:vertexGadget}. Each of the remaining $k$ columns in $I_1(c)$ corresponds to a colour in   $\{1,2,\dots,k\} \setminus\{c\}$.
 These $|V_c(G)|\times (k-1)$ intervals are later paired with intervals from edge-selection gadgets to form 2-intervals that correspond to ``directed'' edges. Notice that the intervals of the first column of $I_1(c)$ pairwise intersect, ensuring that at most one 2-interval corresponding to a vertex with colour $c$ can appear in any feasible solution. Similarly, for the $|V_c(G)|\times (k-1)$ intermediate intervals of $I_1(c)$ (i.e., the intervals of $I_1(c)$ excluding those in the first and last column), it means that all the edges of a $k$-multicoloured clique with at least one endpoint with colour $c$ are incident to the same vertex in $V_c(G)$. 

\begin{figure}[t]
    \centering
	\includegraphics[width=\textwidth]{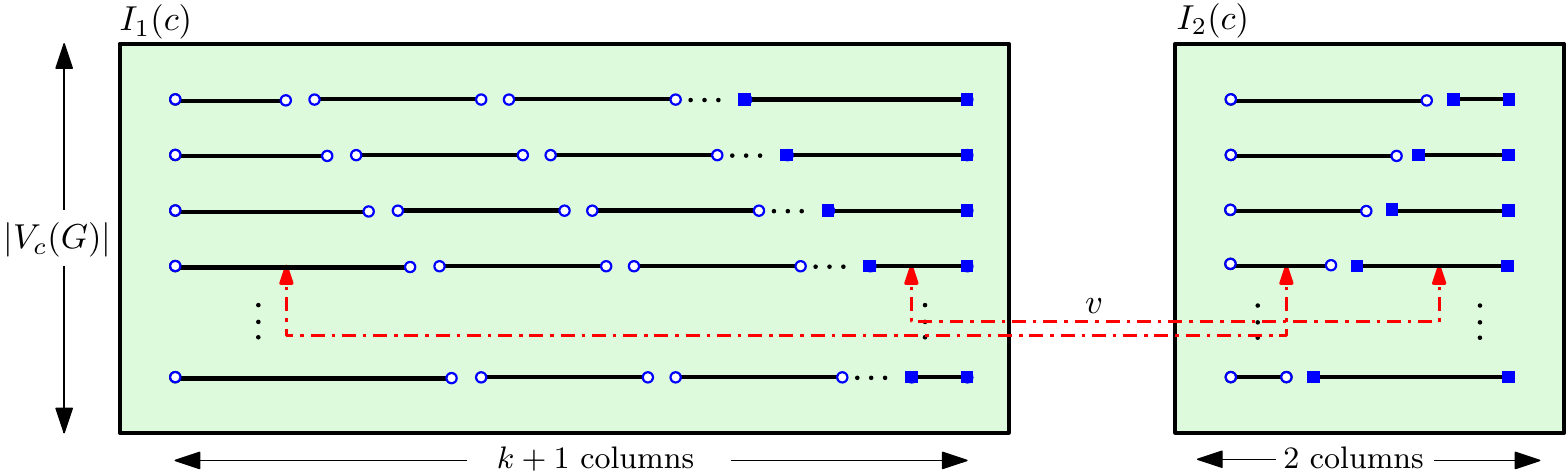}
	\caption{A vertex-selection gadget and the two 2-intervals $I_v$ and $I'_v$ corresponding to a vertex $v$.}
	\label{fig:vertexGadget}
\end{figure}

\begin{lemma}
\label{lem:vertexEnsuring}
Let $S$ be feasible a solution for the \tip, and consider the vertex-selection gadget $T$ corresponding to colour $c$. Moreover, let $M\subseteq S$ be the set of 2-intervals such that each 2-interval in $M$ has at least one interval in $T$. If $|M|\geq k+1$, then all the intervals in $M\cap T$ are selected from the same row of $T$.
\end{lemma}
\begin{proof}
Since there are $(k+1)$ columns in the component $I_1(c)$ of $T$,  $M$ cannot have more than $(k+1)$ 2-intervals, where each containing at least one interval from $T$. Hence, $|M|=k+1$, 
 This means that $M$ must contain exactly one interval from every column of $I_1(c)$ and  hence, one from every column of $I_2(c)$. Consider the interval in the first column of $I_1(c)$ (that is in $M$) and assume that this interval is in row $i$, for some $1\leq i\leq |V_c(G)|$; it corresponds to a vertex $v\in V_c(G)$. We now show that every other interval in $M\cap T$ must also be in row $i$. Since $M\subseteq S$ and $S$ is a feasible solution, then the interval in the first column of $I_2(c)$ and row $i$ must also be in $M$ because these two intervals form one of the two 2-intervals corresponding to $v$. Now, suppose that the interval in $M$ that is from the second column of $I_2(c)$ is in row $i'$. Clearly, $i'\leq i$  (i.e., $i'$ lies below $i$) because otherwise the interval of $M$ that is in the first column of $I_2(c)$ would intersect this interval due to the offset. Since $M\subseteq S$ and $S$ is a feasible solution, $M$ must contain the interval in the last column of $I_1(c)$ that is in row $i'$ (as   only these two would form a valid 2-interval  while considering $I_2(c)$). If $i'<i$, then it is not possible to have exactly one interval from column $j$ of $I_1(c)$ in $M$ for all $j=2,3,\dots,k$ because the offset would imply that at least two intervals must intersect in $M$. Therefore, $i'=i$. In the same way, we can show that the subsequent intervals 
 of $M\cap T$ must also be in row $i$.
\end{proof}

Observe that the assignment of two 2-intervals for each vertex $v\in V_c(G)$ and placement of their second intervals in $I_2(c)$ with an offset allowed us to argue that the remaining intervals are also selected from the same row of the vertex-selection gadget. We will use a similar construction to argue the same for edge-selection gadgets. Before we continue, one might wonder why we needed $I_2(c)$ and why could not we have only $I_1(c)$ with one 2-interval for each vertex. Although this would force the selection of remaining intervals from the same row, it is impossible to place such a gadget on the real line while maintaining $\R$-comparability. To ensure $\R$-comparability, we will need to place $I_1(c)$ and $I_2(c)$ on different parts of the real line, possibly far apart from each other.

\paragraph{Edge-selection gadget.} For each distinct pair of colours $(i,j)$, we construct an edge-selection gadget. The gadget has two main components, which we denote by $C_1(i,j)$ and $C_2(i,j)$. The component $C_1(i,j)$ has $|E_{(i,j)}|$ rows of intervals each of which corresponds to an edge $(u,v)$ of $G$ such that $\{c(u),c(v)\}=\{i,j\}$; see Figure~\ref{fig:edgeGadget}(a). Each row has four columns of intervals; the intervals in the same column pairwise intersect. Moreover, there is an offset such that an interval in column $t$ intersects the interval in column $t+1$ that is in the row immediately above it. The component $C_2(i,j)$ has $|E_{(i,j)}|$ rows and only two columns. There is also an offset between the intervals similar to the offset defined for the intervals in $C_1(i,j)$; see Figure~\ref{fig:edgeGadget}(b). The row $r$ in $C_1(i,j)$ corresponds to an edge $(u,v)\in E(G)$ if and only if the row $r$ in $C_2(i,j)$ corresponds to the edge $(u,v)\in E(G)$.

\begin{figure}[t]
    \centering
	\includegraphics[width=.85\textwidth]{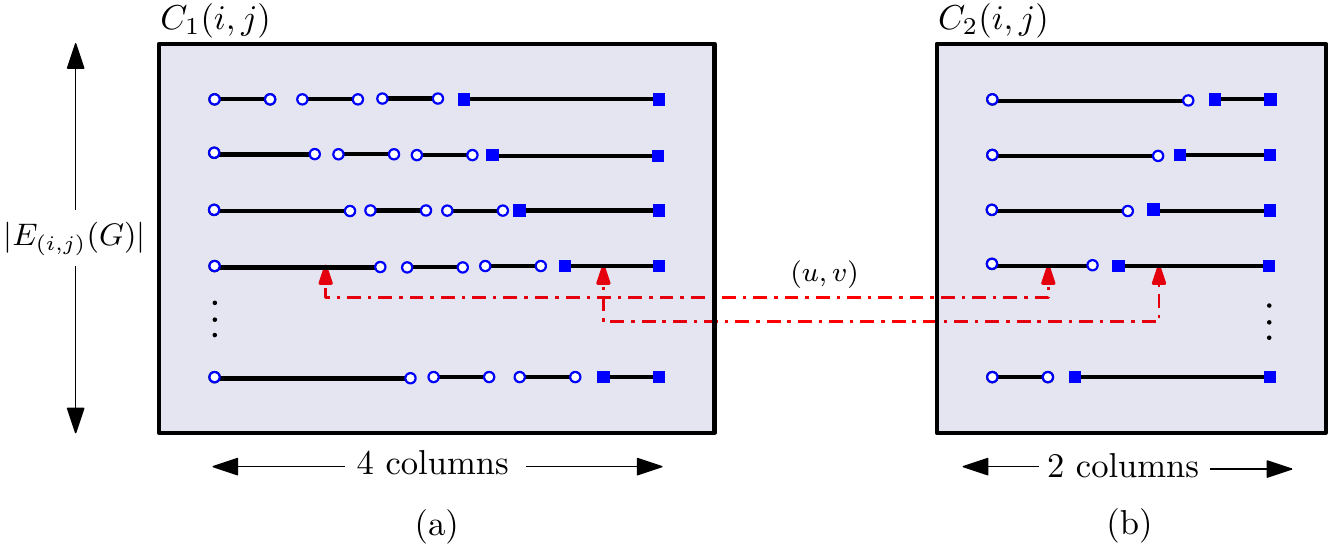}
	\caption{An illustration of the two components of an edge-selection gadget; namely, (a) $C_1(i,j)$ and (b) $C_2(i,j)$. The two 2-intervals $I_{\{u,v\}}$ and $I'_{\{u,v\}}$ corresponding to the edge $(u,v)\in E(G)$ are shown dashed-dotted red.}
	\label{fig:edgeGadget}
\end{figure}

Recall that for each edge in $E(G)$, we associate four 2-intervals; we next describe the construction of these 2-intervals. Let $(u,v)\in E(G)$ such that $c(u)=i$, $c(v)=j$ and $i<j$. Then, the 2-interval $I_{\{u,v\}}$ (resp., $I'_{\{u,v\}}$) is composed of the interval in the first column (resp., last column) of the row corresponding to $(u,v)$ in $C_1(i,j)$ and the first interval (resp., second interval) of the row corresponding to $(u,v)$ in $C_2(i,j)$. See Figure~\ref{fig:edgeGadget} for an illustration. The 2-interval $I_{(u,v)}$ (associated with the ``directed'' edge $(u,v)$) is composed of the interval in the second column of $C_1(i,j)$ and the interval in the vertex-selection gadget of $i$ that is in the row corresponding to vertex $u$ and the column for colour $j$. The 2-interval corresponding to the ``directed'' edge $(v,u)$ is constructed in a similar way: it consists of the interval in the third column of $C_1(i,j)$ and the interval in the vertex-selection gadget of $j$ that is in the row corresponding to vertex $v$ and the column for colour $i$. Figure~\ref{fig:vertexEdgeGadget} illustrates an example for constructing the two 2-intervals corresponding to such ``directed'' edges. Note that the latter two 2-intervals that correspond to ``directed'' edges are used for validation: they ensure that if the 2-intervals of a vertex $u$ with colour $i$ is selected, then all the selected edges with an endpoint of colour $i$ are incident to $u$.
\begin{lemma}
\label{lem:edgeEnsuring}
Let $S$ be a feasible solution for the \tip, and consider an edge-selection gadget $T$. If there are four 2-intervals in $S$ such that each of them has at least one interval in $T$, then all such four 2-intervals must have intervals from the same row of $T$.
\end{lemma}
\begin{proof}
The proof uses an argument similar to the one we used in the proof of Lemma~\ref{lem:vertexEnsuring}. Suppose that $T$ corresponds to edges with colours $i$ and $j$, where $i<j$. Now, consider the gadget $C_1(i,j)$. Clearly, $S$ can have at most one interval from each column of $C_1(i,j)$. Since $S$ has four 2-intervals that have at least one interval in $T$, the set $S$ contains exactly one interval from each column of $T$. Suppose that the interval of the first column of $C_1(i,j)$ (that is in $S$) is at row $t$ for some $1\leq t\leq |E_{(i,j)}|$. Notice that this interval forms a 2-interval with the first interval in row $t$ of $C_2(i,j)$ and so that interval must also be in $S$ (these two intervals form a valid 2-interval and $S$ is a feasible solution). We now show that the interval of the last column of $C_1(i,j)$ (that is in $S$) must also be at row $t$. Suppose for the sake of contradiction that it is at a row $t'\neq t$. First, by construction, this interval forms a 2-interval with the second interval in row $t'$ of $C_2(i,j)$ and so that interval must also be in $S$. If $t'>t$, then the second interval in row $t'$ of $C_2(i,j)$ intersects with the first interval in row $t$ of $C_2(i,j)$ by the construction and so they cannot both be in $S$---a contradiction. Moreover, if $t'<t$, then $S$ cannot contain an interval from both the second and third columns of $C_1(i,j)$ because at least one of them intersects the interval of $S$ that is in either the first or the last column of $C_1(i,j)$---a contradiction. Therefore, $t'=t$ and so the two intervals in $S$ that are in $C_2(i,j)$ are also from the same row $t$. Finally, the fact that $t'=t$ forces the intervals in the second and third columns of $C_1(i,j)$ (that are in $S$) to be also from the row $t$.
\end{proof}

\begin{figure}[t]
	\includegraphics[width=\textwidth]{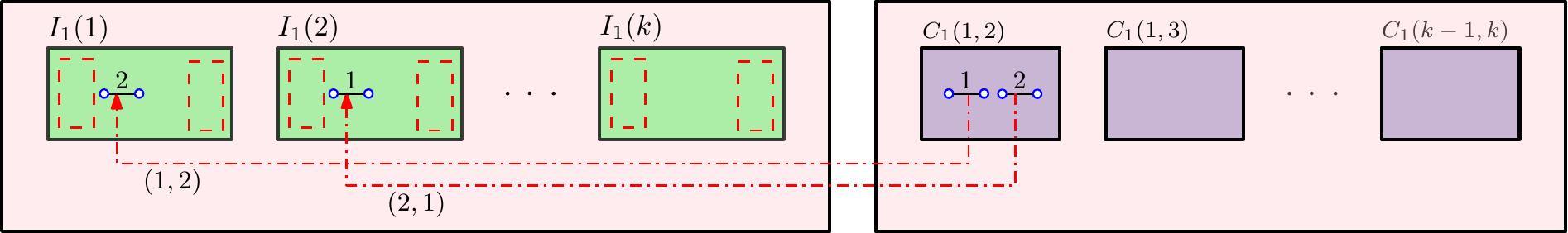}
	\caption{An illustration of the two 2-intervals corresponding to the ``directed'' edges $(u,v)$ and $(v,u)$, assuming $c(u)=1$ and $c(v)=2$. The dashed (red) rectangles shown in gadgets $I_1(\cdot)$ indicate the first and last columns of intervals in the gadget.}
	\label{fig:vertexEdgeGadget}
\end{figure}

By the above constructions, we obtain the set $\F$ of 2-intervals as
\[
\F=\{I_v, I'_v|v\in V(G)\}\cup\{I_{\{u,v\}},I'_{\{u,v\}},I_{(u,v)},I_{(v,u)}|(u,v)\in E(G)\}.
\]
Since we associate each vertex with two 2-intervals and each edge with four 2-intervals, we have $|\F|=2|V(G)|+4|E(G)|$. The construction of our gadgets can all be done in $\fpt$-time. In the following, we show the arrangement of the gadgets on the real line specific to each of $\R=\{\nst,\btw\}$ and $\R=\{<,\btw\}$. Then, we show that any $k$-multicoloured clique in $G$ corresponds to $2k+4{k\choose 2}$ pairwise disjoint 2-intervals of $\F$. For brevity, let $k'=2k+4{k\choose 2}$ for the rest of this section.

\subsection{Hardness for $\R=\{\nst,\btw\}$}
We now show how to arrange the gadgets on the real line when $\R=\{\nst,\btw\}$. To this end, consider the ordering $\{1,2,\dots,k\}$ of colours. We place the gadgets on disjoint regions of the real line from left to right as follows. First, for each pair of distinct colours $i$ and $j$, $1\leq i<j\leq k$, we place the gadget $C_2(i,j)$ on the line in this order; that is, we first place the gadgets $C_2(1,j)$ for all $j=2,\dots,k$, then the gadgets $C_2(2,j)$ for all $j=3,\dots,k$ and so on. Then, we place the gadgets $I_1(c)$ ($1\leq c\leq k$) from left to right in the increasing order of $c$. Next, we place the gadgets $C_1(i,j)$, $1\leq i<j\leq k$ in the same order as we placed their corresponding gadgets $C_2(i,j)$. Finally, we place the gadgets $I_2(c)$ ($1\leq c\leq k$) in the same order as we placed their corresponding gadgets $I_1(c)$. See Figure~\ref{fig:exampleCN} for an example. This forms our instance $(\F,\R, k')$ of the \tip, where $\R=\{\nst,\btw\}$ and $k'=2k+4{k\choose 2}$. Clearly, this arrangement can be done in $\fpt$-time. Moreover, one can verify that any two 2-intervals in this instance are $\R$-comparable, where $\R=\{\nst,\btw\}$.
\begin{lemma}
\label{lem:nestCross}
Graph $G$ has a $k$-multicoloured clique if and only if the \tip on $\F$ has a feasible solution of size $k'$ with respect to $\R=\{\nst,\btw\}$.
\end{lemma}
\begin{proof}
$(\Rightarrow)$ Suppose that $G$ has a $k$-multicoloured clique. For each colour $c$, let $v_c$ be the vertex in the clique with colour $c$. Then, for every colour $c$, we select the two 2-intervals $I_{v_c}$ and $I'_{v_c}$ from the vertex-selection gadget corresponding to $c$. Moreover, for every pair of colours $i$ and $j$ with $i<j$, let $(u_i,u_j)$ be the edge in the clique such that $c(u_i)=i$ and $c(u_j)=j$. Then, we select the four 2-intervals $I_{\{u_i,u_j\}}, I'_{\{u_i,u_j\}}, I_{(u_i,u_j)}$ and $I_{(u_j,u_i)}$. In this way, we have selected $k'$ 2-intervals in total. Moreover, by the arrangement of gadgets on the real line, one can verify that this set of $k'$ 2-intervals is $\R$-comparable.

\begin{figure}[t]
	\includegraphics[width=\textwidth]{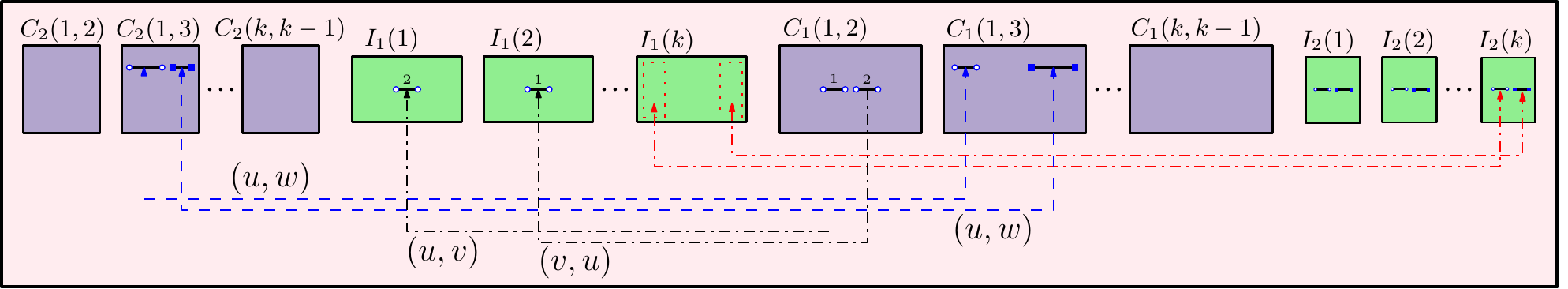}
	\caption{The arrangement of gadgets for $\R=\{\nst,\btw\}$.} 
	\label{fig:exampleCN}
\end{figure}

$(\Leftarrow)$ Consider a set $S$ of $k'$ 2-intervals that is a feasible solution for the \tip with respect to $\R=\{\nst,\btw\}$. First, observe that $S$ can have at most one interval from the first column of every vertex-selection gadget. We now show that it must contain at least one such interval from the first column of every vertex-selection gadget. Let $S_1\subseteq S$ (resp., $S_2\subseteq S$) be the set of 2-intervals such that each 2-interval in $S_1$ has at least one interval in a vertex-selection gadget (resp., an edge-selection gadget). Moreover, let $S_3\subseteq S$ (resp., $S_4\subseteq S$) be the set of 2-intervals such that each 2-interval in $S_3$ (resp., $S_4$) has exactly two intervals from the same vertex-selection gadget (resp., the same edge-selection gadget). Observe that $|S_2|\leq 4{k\choose 2}$ because the component $C_1(\cdot)$ of an edge-selection gadget has four columns and no two intervals in $S$ can come from the same column of any given $C_1(\cdot)$. This means that $|S_3|\geq 2k$. But, there are exactly $k$ vertex-selection gadgets and at most two 2-intervals of $S_3$ can be from the same vertex-selection gadget. Hence, $|S_3|=2k$ and so $|S_2|=4{k\choose 2}$. Since there are exactly ${k\choose 2}$ edge-selection gadgets, it follows that we have exactly four 2-intervals in $S$ that come from the same edge-selection gadget. By Lemma~\ref{lem:edgeEnsuring}, all the 2-intervals coming form the same edge-selection gadget lie in the same row of the gadget.

On the other hand, $|S_1\setminus S_3|\leq k(k-1)=2{k\choose 2}$ because we have $k$ vertex-selection gadgets, the component $I_1(\cdot)$ of any vertex-selection gadget has $k-1$ ``internal'' intervals, and at most one of such internal intervals (per column, per vertex-selection gadget) can be in $S$. Notice that a 2-interval has exactly one interval in a vertex-selection gadget if and only if it has exactly one interval in an edge-selection gadget. Therefore, $S_2\setminus S_4=S_1\setminus S_3$. Since $\{S_1,S_3,S_4\}$ (or, $S_2,S_3,S_4$) forms a partition of $S$, we must have $|S_2\setminus S_4|=|S_1\setminus S_3|=2{k\choose 2}$. That is, there are exactly $2{k\choose 2}$ 2-intervals that have exactly one interval in a vertex-selection gadget and the other interval in an edge-selection gadget. Notice that at most $k-1$ of such $2{k\choose 2}$ 2-intervals can come from the same vertex-selection gadget. Since there are $k$ vertex-selection gadgets, there are exactly $k-1$ of them from each vertex-selection gadget. This means that, for each vertex-selection gadget, there are $k+1$ 2-intervals in $S$ that come from this gadget. By Lemma~\ref{lem:vertexEnsuring}, these $k+1$ 2-intervals all come from the same row of the gadget. Hence, we select the $k$ vertices corresponding to these $k$ rows. We now claim that they are a feasible solution for the $k$-multicoloured clique. Clearly, each selected vertex has a unique colour. Moreover, take any colour $c$ and let $u$ be the vertex that we selected with colour $c$. Recall that all the intervals of $S$ that come from the vertex-selection gadget $c$ are in the same row as that of $u$. There are $k-1$ of them (excluding those corresponding to $u$ itself) and each is paired with an interval in an edge-selection gadget corresponding to the pair $(c,c')$ of colours, for all colours $c'\neq c$. Therefore, there exists an edge between $u$ and every other selected vertex and so the $k$ selected vertices are indeed a feasible solution for the $k$-multicoloured clique.
\end{proof}

By Lemma~\ref{lem:nestCross}, we have the following theorem.
\begin{theorem}
The \tip is $\W$-hard when $\R=\{\nst,\btw\}$.
\end{theorem}

\subsection{Hardness for $\R=\{<,\btw\}$}
We now show that the \tip is $\W$-hard even when $\R=\{<,\btw\}$. To this end, we show how to arrange the gadgets on the real line such that any pair of two 2-intervals are $\{<,\btw\}$-comparable. Then, one can prove a result similar to Lemma~\ref{lem:nestCross} for $\R=\{<,\btw\}$, concluding that the problem is $\W$-hard even for $\R=\{<,\btw\}$. Here, we only show the arrangement.

\begin{figure}[t]
	\includegraphics[width=\textwidth]{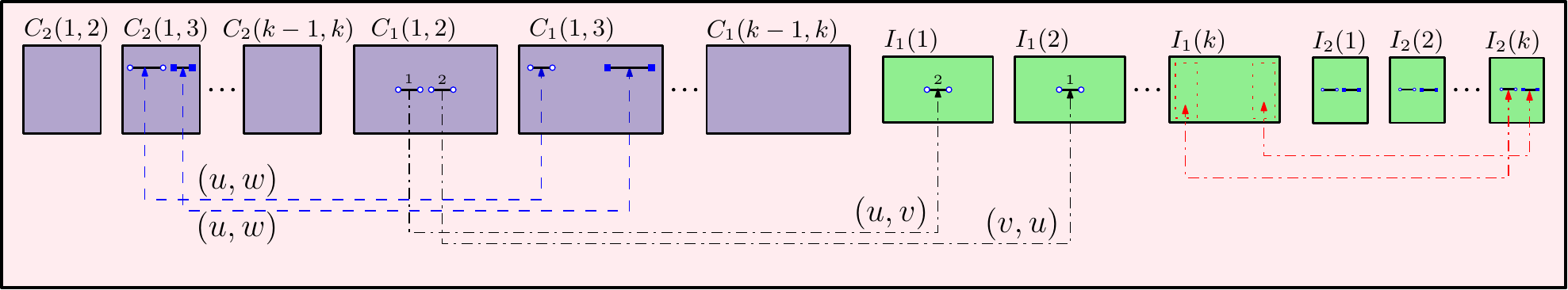}
	\caption{The arrangement of gadgets for $\R=\{<,\btw\}$.}
	\label{fig:exampleCP}
\end{figure}

Consider the ordering $\{1,2,\dots,k\}$ of colours. We place the gadgets on disjoint regions of the real line from left to right as follows. First, for each pair of distinct colours $i$ and $j$, $1\leq i<j\leq k$, we place the gadget $C_2(i,j)$ on the line in this order; that is, we first place the gadgets $C_2(1,j)$ for all $j=2,\dots,k$, then the gadgets $C_2(2,j)$ for all $j=3,\dots,k$ and so on. Then, we place the gadgets $C_1(i,j)$, $1\leq i<j\leq k$ in the same order as we placed their corresponding gadgets $C_2(i,j)$. Next, we place the gadgets $I_1(c)$ ($1\leq c\leq k$) from left to right in the increasing order of $c$. Finally, we place the gadgets $I_2(c)$ ($1\leq c\leq k$) in the same order as we placed their corresponding gadgets $I_1(c)$. See Figure~\ref{fig:exampleCP} for an example. This forms our instance $(\F,\R, k')$ of the \tip, where $\R=\{<,\btw\}$ and $k'=2k+4{k\choose 2}$. Clearly, this arrangement can be done in $\fpt$-time and one can verify that every of pair of 2-intervals are $\{<,\btw\}$-comparable.
\begin{theorem}
The \tip is $\W$-hard when $\R=\{<,\btw\}$.
\end{theorem}

\section{Conclusion}
\label{sec:conclusion}
In this paper, we showed that the \tip is $\W$-hard when $\R=\{\nst,\btw\}$ and $\R=\{<,\btw\}$; hence, fully settling the parameterized complexity of the problem when parameterized by the size of an optimal solution. It would be interesting to examine $\fpt$-algorithms with respect to other parameters such as the maximum number of pairwise intersecting $2$-intervals.

\bibliographystyle{plain}
\bibliography{ref}

\end{document}